\documentclass [12pt]{article}

\usepackage{amsmath,amsthm,amscd,amssymb}

\usepackage[small, centerlast, it]{caption}
\usepackage{epsfig}
\usepackage[titles]{tocloft}
\usepackage[latin1]{inputenc}
\usepackage{verbatim} 
\usepackage[active]{srcltx} 
\usepackage{fancyhdr}
\usepackage{caption}

\def\b1{{1\!\!1}}

\def\sH{{\mathsf H}}

\def\bC{{\mathbb C}}           

\def\bI{{\mathbb I}}

\def\bN{{\mathbb N}}

\def\bR{{\mathbb R}}

\def\gB{{\mathfrak B}}

\def\gS{{\mathfrak S}}

\def\beq{\begin{eqnarray}}
\def\eeq{\end{eqnarray}}

\usepackage{sectsty}
\sectionfont{\normalsize}
\subsectionfont{\normalsize\normalfont\itshape}
\newtheoremstyle{thm}
{12pt}
{12pt}
{\itshape}
{}
{\itshape\bfseries}
{}
{1em}
{}

\theoremstyle{thm}
\newtheorem{theorem}{Theorem}

\newtheorem{proposition}[theorem]{Proposition}
\newtheorem{definition}[theorem]{Definition}

\title{\Large A quantum key distribution scheme based on tripartite entanglement and violation of CHSH inequality}
\author{ Davide Pastorello\footnote{d.pastorello@unitn.it}\\\normalsize Department of  Mathematics, University of Trento,\\
\normalsize Trento Institute for Fundamental Physics and Applications\\
 \normalsize via Sommarive 14, 38123 Povo (Trento), Italy.}

\date{}

\input{Qcircuit}

\begin{document}



\maketitle
\begin{abstract}

\noindent
Entanglement is a well-known resource in quantum information, in particular it can be exploited for quantum key distribution (QKD). 
 In this paper we define a two-way QKD scheme employing GHZ-type states of three qubits obtaining an extension of the standard E91 protocol with a significant increasing of the number of shared bits. Eavesdropping attacks can be detected measuring violation of the CHSH inequality and the secret key rate can be estimated in a device-independent scenario. 

\end{abstract}

\small{Keywords: \emph{Quantum key distribution, entanglement, Bell's inequalities }}

\vspace{0.5cm}	

\section{Introduction}

Quantum processes can be used to distribute a secret key (a random bit-string) in order to use it as a one-time pad for secure communications \cite{n,m}. 
Quantum key distribution (QKD) over a quantum channel prevents effective eavesdropping attacks exploiting principles of Quantum Mechanics. An eavesdropper cannot clone an unknown quantum state (\emph{no-cloning theorem} \cite{no-cloning}) thus he extracts information performing measurement processes, the quantum effects of these measurements can be detected showing that the communication channel is not secure. One of the most celebrated QKD scheme is the E91 protocol \cite{e} which exploits quantum correlations of entangled states to share a private key and uses a measure of the violation of Bell inequalities to check if an eavesdropping occured over the quantum channel. 
\\This paper is devoted to define and analyze a modification of standard E91 protocol based on preparation and processing of tripartite entangled states. The main idea is defining a two-way protocol where Bob transmits a half of an entangled state to Alice who entangles the received state with a third part, as an encoding opeartion, before re-sending it to Bob. Then he perfoms a decoding to read the information communicated by Alice, after the decoding operation the clients share a maximally entangled state which can be used to implement a standard E91 protocol as a subroutine. The classical postprocessing is equivalent to that of E91 protocol, so the security analysis is similar. The advantage w.r.t. to the standard E91 protocol (which can be described as a one-way protocol) is the doubled length of the secret key at the same cost of transmitted qubits. Even if an eavesdropper, Eve, can attack each signal twice, we observe that any Eve's strategy produces a destruction of quantum correlations which can be taken under control by the violation of CHSH inequality. In particular the effect of noise on the secret key rate can be evaluated in terms of such a violation.
\\
In the next section we recall some basic features of quantum systems that are relevant in the present discussion, in particular we refer to \emph{finite-dimensional case}.
In the third section there is a short review on Bell's inequalities, their physical meaning and relative application in quantum key distribution (E91 protocol). In the fourth section we describe the proposed QKD scheme as a two-way extension of the standard E91 protocol based on tripartite entanglement. 
 In the fifth section we analyze the robustness of the protocol against eavesdropping like \emph{intercept/re-send} attacks from a qualitative viewpoint. We show that the security of the protocol is based on the violation of CHSH inequality and entanglement monogamy. In the sixth section we remark that an extimation of the secret key rate can be computed in a device-independent scenario after applying a prescription to purify two-way QKD protocols.  In the last section thare are conclusions and an overview on open issues about the security proof. A technical observation about the entanglement class of the considered states is given in appendix.

\section{Preliminaries and notations}

According to standard formulation of Quantum Mechanics, a complex vector space $\sH$ with inner product $\langle\,\,\,|\,\,\,\rangle$ (a Hilbert space) is associated to any quantum system\footnote{We adopt \emph{Dirac formalism}: A unit vector of the Hilbert space is denoted by the ket $|\psi\rangle$, a vector of dual space is denoted by the bra $\langle\psi|$, the inner product of two vectors is denoted by $\langle\psi|\phi\rangle$, the outer product by $|\psi\rangle\langle\phi|$.}. If $\dim\sH=2$ then the considered quantum system is called \emph{qubit}.   
\\Let $\gB(\sH)$ be the space of linear operators on $\sH$, the set of physical states of the considered quantum system is defined as:
\beq
\gS(\sH):=\{\rho\in\gB(\sH): \rho\geq 0,\, \mbox{tr}\rho=1\}.
\eeq
$\gS(\sH)$ is convex in $\gB(\sH)$ and its extremal elements, called \emph{pure states}, are the rank-1 orthogonal projectors in $\sH$, a non-pure state is called \emph{mixed state}. Since a pure state $\rho$ can be written as $\rho=\ket\psi\bra\psi$ for some unit vector $\ket\psi\in\sH$, then pure states can be represented by equivalence classes of unit vectors where the equivalence relation between $\ket\psi$ and $\ket{\psi'}$ is defined by $\ket\psi\sim\ket{\psi '}$ iff $\ket\psi=e^{i\theta}\ket{\psi'}$ for some $\theta\in\bR$. Therefore two unit vectors differing by a multiplicative phase factor describe the same pure state.  \\
Time evolution of an isolated system is described by a continuous one-parameter group of unitary operators $\{U(t)\}_{t\in\bR^+}$ acting on states. If $\rho_1$ is the state of the system at time $t_1$ and $\rho_2$ is the state of the system at time $t_2>t_1$ then:
\beq
\rho_2=U(t_2-t_1)\rho_1 U^*(t_2-t_1).
\eeq
A measurement process on a quantum system is described by a collection of positive operators $\{E_k\}$ satisfying  $\sum_k E_k=\b1_\sH$ called \emph{positive operator-valued measure} (POVM), the index $k$ runs in the set of all possible outcomes of the measurement, so it is a real number. The probability to measure $k$ when the system is in the state $\rho$ is:
 \beq
p_\rho(k)=\mbox{tr}(E_k\rho).
\eeq 
If we consider a pure state $\ket\psi$ then above probability is simply given by $p_\psi(k)=\bra\psi E_k \psi\rangle.$

\noindent
A special class of POVMs is that of \emph{projective valued measures} (PVMs) whose elements are orthogonal projectors $P_k$, i.e. operators on $\sH$ satisfying $P_k^*=P_k$ and $P_k^2=P_k$. In this case if the measurement is performed when the state of the system is $\rho\in\gS(\sH)$ and it produces the outcome $k$ then the state of the system after the measurement is: 
\beq
\rho'=\frac{P_k\rho P_k}{\mbox{tr}(P_k\rho)},
\eeq
and a selfadjoint operator on $\sH$, called \emph{observable}, can be defined:
\beq
A:=\sum_k kP_k,
\eeq
so the possible outcomes of the measurement described by the PVM $\{P_k\}$ are the eigenvalues of $A$. If the state of the system is $\rho\in\gS(\sH)$ then the \emph{expectation value} of $A$ is 
\beq
\langle A\rangle_\rho=\sum_k k\,\, p_\rho(k)=\mbox{tr} (A\rho).
\eeq

\noindent
Another fundamental quantum feature is the following: If a quantum system is composed by two subsystems $A$ and $B$ that are respectively described in
Hilbert spaces $\sH_A$ and $\sH_B$, then the composite system is described in the Hilbert space tensor product $\sH_A\otimes \sH_B$.\\
Quantum systems, in particular qubits, can be exploited as information storages and physical operations on them can be used for processing and transmitting information.
Let $\sH$ be a 2-dimensional Hilbert space and $(\ket 0, \ket1)$ be an orthonormal basis of $\sH$ that can be identified as the standard basis of $\bC^2$.  In the following we call $(\ket 0, \ket1)$ the \emph{computational basis} of $\sH$. One can assume the vectors $\ket 0$ and $\ket 1$ physically represent the eigenstates of a reference observable, e.g. the rectilinear polarization of a photon described by the Pauli matrix $\sigma_z$. Another basis of $\sH$ is given by $(\ket +, \ket -)$ with $\ket + :=1/\sqrt {2}(\ket 0+\ket 1)$ and $\ket- :=1/\sqrt{2}(\ket 0-\ket 1)$, the vectors $\ket +$, $\ket -$ can be physically interpreted as the eigenstates of circular polarization of a photon for instance.

\begin{definition}
Let $\sH$ be a 2-dimensional Hilbert space and $n\in\bN$. The Hilbert space $\sH^{\otimes n}$ is called \textbf{n-qubit register} and any unitary operator acting on $\sH^{\otimes n}$ is called \textbf{n-qubit quantum gate}.
\end{definition}
\noindent
The \emph{Hadamard gate} $H$ is a 1-qubit quantum gate defined w.r.t. the computational basis by $H|0\rangle:=|+\rangle$ and $H|1\rangle:=|-\rangle$
and it is denoted by the symbol:
\beq
\Qcircuit @C=1cm @R=2cm {
& \gate{H} & \qw
}
\eeq
The Hilbert space of a qubit pair is $\sH\otimes\sH$, with $\dim\sH=2$, an orthonormal basis of $\sH\otimes\sH$ is $\left(\ket{00},\ket{01},\ket{10},\ket{11}\right)$, where $\ket {00}\equiv \ket 0 \otimes \ket 0$, that we call \emph{computational basis} of $\sH\otimes\sH$. A remarkable 2-qubit gate is the \emph{CNOT gate} $\Lambda$ that is defined in the computational basis by:
\beq
\Lambda |00\rangle:=|00\rangle\quad \Lambda |01\rangle:=|01\rangle\quad\Lambda |10\rangle:=|11\rangle\quad \Lambda |11\rangle:=|10\rangle,
\eeq
i.e. $\Lambda$ acts as the identity when the first qubit is in $|0\rangle$ and acts as a bit-flip on the second qubit when first qubit is in $|1\rangle$. The CNOT gate is denoted by the symbol:
\beq
\Qcircuit @C=2em @R=1.7em {
& \ctrl{1} & \qw& \\
& \targ & \qw &
}
\eeq
\\
CNOT gate can be used to define another useful element of a quantum circuit, the \emph{SWAP gate}:
$$
\Qcircuit @C=2em @R=2.8em {
& \qswap & \qw& \\
& \qswap\qwx & \qw &
} 
\Qcircuit @C=2em @R=1.5em {
&  & \\
& \lstick{:=}&\\
&  &
}
\Qcircuit @C=2em @R=1.7em {
& \targ&\ctrl{1} &\targ&\qw \\
& \ctrl{-1}&\targ &\ctrl{-1}&\qw
}
$$
that switches the input qubits. In the fourth section we define a quantum circuit to entangle three qubits for our QKD purpose.

\vspace{0.0cm}

\section{CHSH inequality and standard E91 protocol}

\noindent
Let us briefly recall the notions of entanglement, violation of Bell's inequalities and E91 protocol. Since composite quantum systems are described in tensor product Hilbert spaces, the definition of quantum state itself implies that there is a class of quantum states containing non-classical correlations that are inherently non-local.
\begin{definition}\label{entanglement}
Let $\rho\in\gS(\sH_A\otimes\sH_B)$ be a state of a composite quantum system. $\rho$ is said to be \textbf{separable} if it can be written as: 
\beq
\rho=\sum_{i} \lambda_i \rho_i^{A}\otimes\rho_i^{B}
\eeq
with $\rho_i^{A}\in\gS(\sH_A)$, $\rho_i^{B}\in\gS(\sH_B)$,  $\lambda_i\geq 0$ and $\sum_i\lambda_i=1$. Otherwise it is called \textbf{entangled}.
\end{definition}

\noindent
Entanglement can be interpreted as a \emph{non-local} property: Consider the entangled pure state of a qubit pair identified by the vector $\ket\Phi=\frac{1}{\sqrt{2}}(\ket 0 \otimes \ket\psi + \ket 1\otimes \ket\varphi)\in\sH_A\otimes\sH_B$, if one performs the PVM-measurement $\{\ket 0\bra 0, \ket 1\bra 1\}$ on the first qubit then the state of the second qubit collapses in $\ket\psi$ or in $\ket\varphi$ according to the outcome of the local measurement, this phenomenon occurs even if the qubits are spacelike separated (\emph{EPR paradox} \cite{epr}). However this non-local feature cannot be used for communications faster than light (\emph{no-communication theorem} \cite{peres}). 
\\
In order to elucidate the content of \emph{Bell theorem} \cite{bell}, let us briefly discuss the notion of \emph{local hidden variable theory} based on the assumption of {local realism}. Let us adopt the term \emph{local  realism} in reference to a pair of properties that a general physical theory can present: 
\\
\\
i) \emph{Principle of locality}: If two events are outside their respective light cones then there is no causal connection among them. 
\\
ii) \emph{Counterfactual-definiteness}: The values of physical quantities are definite and unaffected by any measurement process. 
\\
\\
Consider a bipartite quantum system composed by a subsystem $A$ and a subsystem $B$, suppose to perform local measurements on $A$ and $B$ obtaining correlated outcomes when they are spatially separated. Invoking local realism we are assuming that local measurements are independent and do not perturb the systems. So the probability to measure value $k\in\bR$ by a POVM measurement $E$ on the system $A$ and value $l\in\bR$ by a POVM measurement $D$ on the system $B$ can be generally expressed by:
\beq
p(E,k;D,l)=\int_X F_E(k,x)F_D(l,x)d\mu(x),
\eeq
where $(X,\mu)$ is a measure space, $x\mapsto F_E(k,x)$ and $x\mapsto F_D(l,x)$ are measurable functions associated to POVMs $E$ and $D$. The product $F_E(k,x)F_D(l,x)$ is the probability to obtain value $k\in\bR$ by the measurement process $E$ on the first subsystem and value $l\in\bR$ by the measurement process $D$ on the second subsystem for a fixed value of the parameter $x$, called \emph{hidden variable}, encoding correlations between subsystems. 
\begin{definition}
A quantum state $\rho\in\gS(\sH_A\otimes\sH_B)$ admits a \textbf{hidden variable model} if for all POVMs $E=\{E_k\}_{k}\subset\gB(\sH_A)$ and $D=\{D_l\}_{l}\subset\gB(\sH_B)$ there exists a measure space $(X,\mu)$ and measurable functions  $x\mapsto F_E(k,x)$ and $x\mapsto F_D(l,x)$  such that:
\beq
\emph{tr}(\rho(E_k\otimes D_l))=\int_X F_E(k,x)F_D(l,x)d\mu(x).
\eeq
for any $k$ and $l$.  
\end{definition}

\noindent
The above definition provides a class of states of a composite quantum system such that introducing a fictious variable $x$ quantum correlations can be treated as classical correlations in presence of {local realism}.
 The set of states admitting a hidden variable model is non-empty, in particular it contains all separable states.

\begin{proposition}\label{sepBell}
Any separable state $\rho\in\gS(\sH_A\otimes\sH_B)$ admits a hidden variable model.
\end{proposition}  

\begin{proof}
If $\rho$ is separable then we have the convex combination:
$$\rho=\sum_{i=1}^n \lambda_i \rho_i^{(A)}\otimes\rho_i^{(B)}$$
with $\rho_i^{(A)}\in\gS(\sH_A)$, $\rho_i^{(B)}\in\gS(\sH_B)$. Let be $X:=\{1,2,3,...,n\}$ and $\mu(\{i\}):=\lambda_i$. For all POVMs $\{E_k\}_k\in\gB(\sH_A)$ and $\{D_l\}_l\in\gB(\sH_B)$ let us define:
$$F_E(k,x):=\rho_x^{(A)}(E_k)\quad ,\quad F_D(l,x):=\rho_x^{(B)}(D_l),$$
obtaining the hidden variable model.
\end{proof}

\noindent
Within a hidden variable model, correlation functions obey to a set of constraints called \textbf{Bell's inequalities}. In particular we consider the \emph{CHSH-inequality} (Clauser, Horne, Shimony, Holt \cite{CHSH}) to formulate the Bell theorem:
\begin{theorem}\label{CHSH}
If a state $\rho\in\gS(\sH_A\otimes\sH_B)$ admits a hidden variable model then it satisfies CHSH-inequality:
\beq\label{chsh}
|\emph{tr}[\rho(A\otimes(B-B'))]+\emph{tr}[\rho(A'\otimes (B+B'))]|\leq 2,
\eeq
for all $A,A'\in\gB(\sH_A)$ such that $-\b1_{\sH_A}\leq A,A'\leq \b1_{\sH_A}$ and for all $B,B'\in\gB(\sH_B)$ such that $-\b1_{\sH_B}\leq B,B'\leq \b1_{\sH_B}$.
\end{theorem}
 
\noindent
We have that the violation of CHSH inequality is a sufficient condition for the entanglement of a state, on the countrary there are entangled states admitting a hidden variable model \cite{werner}. 
A remarkable example of CHSH violation is the following: Consider the Pauli matrices $\sigma_x$ and $\sigma_z$ and define the hermitian operators in $\gB(\bC^2)$:
$
A:=\sigma_x, A':=\sigma_z,  B:=\frac{1}{\sqrt 2}(\sigma_x+\sigma_z),  B':=\frac{1}{\sqrt 2}(-\sigma_x+\sigma_z).
$
Consider the unit vector in $\bC^2\otimes\bC^2$:
\beq\label{Bell1}
\ket{\Phi}:=\frac{1}{\sqrt 2} (\ket{00}+\ket{11}),
\eeq
i.e. an entangled pure state of a qubit pair. We can calculate the left-hand term of (\ref{chsh}) where $\rho=\ket\Phi\bra\Phi$, obtaining by direct inspection:
\beq\label{violation}
|\bra{\Phi}A\otimes(B-B')\Phi\rangle+\bra{\Phi}A'\otimes (B+B')\Phi\rangle|=2\sqrt 2,
\eeq
then CHSH inequality is violated. More precisely (\ref{Bell1}) is a maximally entangled state and (\ref{violation}) is the maximum violation of CHSH inequality (experimentally measured for the first time by Aspect et al.  in 1982 \cite{aspect}). 
\\
The E91 protocol \cite{e}  is based on the production of a qubit pair in the maximally entangled pure state (\ref{Bell1}) by Alice. Then she sends one qubit of the pair to Bob\footnote{The entangled qubit pair can be produced by an external source (including an eavesdropper) sending one qubit to Alice and the other one to Bob.}. Alice measures one of the following observables on her qubit:
\beq\label{A}
A_1=\sigma_x\qquad A_2=\frac{1}{\sqrt{2}}(\sigma_x+\sigma_z) \qquad A_3=\sigma_z,
\eeq
on the other side Bob measures one of the following observables on his qubit: 
\beq \label{B}
B_1=\frac{1}{\sqrt{2}}(\sigma_x+\sigma_z) \qquad B_2=\sigma_z \qquad B_3=\frac{1}{\sqrt{2}}(-\sigma_x+\sigma_z).
\eeq
 If qubits are polarized photons or spin-$\frac{1}{2}$ particles then different measurement processes correspond to different horientation angles of Alice and Bob's analyzers.
If Alice and Bob perform the same measurement process then their outcomes coincide.
\\
Let us summarize  E91 protocol as follows:
\\
\\
\textbf{Step 1}: A pair of entangled qubits is produced by Alice. She sends a qubit of the entangled pair to Bob. Alice measures observable $A_i$ with $i\in\{1,2,3\}$ on her qubit. Bob measures observable $B_j$ with $j\in\{1,2,3\}$ on his qubit. The procedure is repeated $N$ times.
\\
\textbf{Step 2}: Alice and Bob declare on a classical public channel what observable they have measured on each qubit. They store the outcomes produced measuring the same observable obtaining two perfectly correlated bit string (sifted key) of length $n\simeq \frac{N}{3}$.
\\
\textbf{Step 3}: Alice and Bob declare on a classical public channel the outcomes produced by different measurements in order to calcuate:
\beq\label{S}
\mathcal S:=|\bra{\Phi}A_1\otimes(B_2-B_3)\Phi\rangle+\bra{\Phi}A_3\otimes (B_2+B_3)\Phi\rangle|.
\eeq
\noindent
At the end of the round the maximum violation of CHSH inequality is expected $\mathcal S=2\sqrt2$. If the experimental value is $\mathcal S<2\sqrt 2$ then some measurement process in the middle of quantum transmission has alterated the state of entangled pairs destroying the quantum correlation, hence an eavesdropping attack has been detected. However a sifted key may be not discarded even if $\mathcal S<2\sqrt 2$, since an estimation of Eve's information in post-processing could show that the survived quantum correlations allowed to produced an acceptable secret key. In this regards see section 6 where some estimations are sketched for the two-way protocol. On the other hand the extreme case $\mathcal S\leq 2$ entails the ereasure of any quantum correlation.\\
Obviously a classical authentication protocol must be implemented over the classical channel in order to ensure that the right person
is at the end of the line. 

\vspace{0.cm}

\section{The proposed QKD protocol}

According to definition \ref{entanglement} a pure state $\ket\phi\bra\phi \in\gS(\sH_A\otimes\sH_B)$ is separable if and only if $\ket\phi=\ket{\psi_A}\otimes\ket{\psi_B}$ with $\ket{\psi_A}\in\sH_A$ and $\ket{\psi_B}\in\sH_B$. More generally: Let $\{\sH_1,...,\sH_n\}$ be a family of finite-dimensional Hilbert spaces. A pure state, represented by $\ket\phi\in\sH_1\otimes\cdots\otimes\sH_n$, is {separable} if and only if $\ket\phi=\ket{\psi_1}\otimes\cdots\otimes \ket{\psi_n}$ with $\ket{\psi_i}\in\sH_i$ $i=1,...,n$. Otherwise if $\ket\phi$ is not a product vector then the corresponding pure state is {entangled}. 

\noindent
For instance a pure state of three qubits can be separable like $\ket{010}$,  \emph{2-entangled} like $\ket\Psi=\frac{1}{\sqrt2}(\ket{01}+\ket{10})\otimes\ket 1$ or \emph{3-entangled} like $\ket{GHZ}=\frac{1}{\sqrt 2}(\ket{000}+\ket{111})$. More precisely pure states of three qubits $a,b,c$ are divided in six entanglement classes: One class of separable states ($a$-$b$-$c$), three classes of 2-entangled states ($a$-$bc$, $b$-$ac$, $c$-$ab$) and two classes of 3-entangled states. These classes correspond to the orbits of the group $GL(2, \bC)^{\otimes 3}$ 
 \cite{cirac}. Representative elements of inequivalent classes of genuine tripartite entangled states of three qubits are:

\beq
\ket{GHZ}=\frac{1}{\sqrt 2}(\ket{000}+\ket{111})\qquad\mbox{and} \qquad \ket W =\frac{1}{\sqrt 3}(\ket{100}+\ket{010}+\ket{001}).
\eeq                        
\\
In this section we propose a QKD protocol where an entangled state of GHZ-type is prepared and processed in order to implement an improved version of E91 protocol which is capable of a \emph{superdense key distribution}, i.e. the length of the sifted key is double at the same cost of standard E91 protocol in terms of the transmitted qubits number.  
\\
Consider the Hilbert space $\sH\otimes\sH$ of a qubit pair, an orthonormal basis can be defined by means of the Bell states:    
\beq
\ket{\Phi^+}=\frac{1}{\sqrt 2}\left (\ket {00}+\ket{11}\right)\,\, ,\quad \ket{\Phi^-}=\frac{1}{\sqrt 2}\left (\ket {00}-\ket{11}\right),
\eeq
$$\ket{\Psi^+}=\frac{1}{\sqrt 2}\left (\ket {01}+\ket{10}\right)\,\,,\quad \ket{\Psi^-}=\frac{1}{\sqrt 2}\left (\ket {01}-\ket{10}\right),\quad\,\,\,$$
\\
representing the maximally entangled states of two qubits.
\\
Let $\Sigma:\sH\otimes\sH\rightarrow\sH\otimes\sH$ be a quantum gate defined in the computational basis as follows:
\beq
\Sigma\ket{00}:=\ket{\Phi^+}\quad,\quad\Sigma\ket{11}:=\ket{\Phi^-}\quad ,\quad \Sigma\ket{10}:=\ket{\Psi^+}\quad,\quad \Sigma\ket{01}:=\ket{\Psi^-}.
\eeq
The corresponding graphical represenation is:
\beq
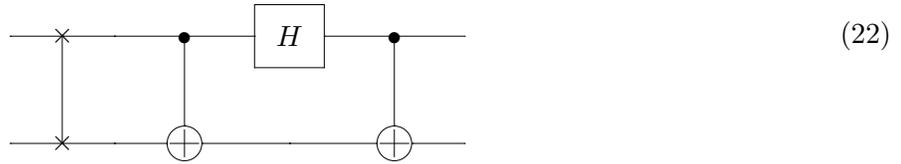
\label{sigma}
\Qcircuit @C=1.8em @R=2.em {
& \qswap & \qw&\ctrl{1}&\gate{H}&\ctrl{1}&\qw&\\
& \qswap\qwx & \qw & \targ&\qw&\targ&\qw&
} 
\eeq

\vspace{0.5cm}
\noindent
The simplest quantum gate to construct Bell states is
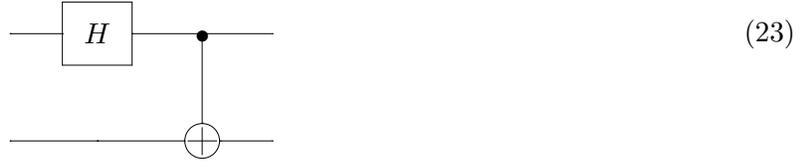
\beq\label{bellgate}
\Qcircuit @C=1.8em @R=2.em {
&\gate{H}&\ctrl{1}&\qw&\\
&\qw&\targ&\qw&
} 
\eeq

\vspace{0.3cm}
\noindent
whose action on the computational basis is $\ket{00}\mapsto\ket{\Phi^+}$, $\ket{11}\mapsto\ket{\Psi^-}$, $\ket{10}\mapsto\ket{\Phi^-}$ $\ket{01}\mapsto\ket{\Psi^+}$.
However $\Sigma$ admits an elementary property that we are going to apply.
\begin{proposition}
Let $\Omega:\sH^{\otimes 3}\rightarrow\sH^{\otimes 3}$ be the quantum gate defined as $\Omega:=(\bI\otimes\Sigma)(\Sigma\otimes\bI)$. Therefore $\Omega \left(\ket\psi \otimes \ket{\Phi^+}\right)=\,\,\,\ket{\Phi^+}\otimes \ket\psi$
for any $\ket\psi\in\sH$.
\end{proposition} 
\vspace{0.0cm}
\begin{proof}
Consider $\ket \psi\in\sH$ decomposed on the computational basis $\ket\psi=a\ket 0 +b\ket 1$, $a,b\in\bC$. One can explicitely calculate the action of $\Omega$ on $\ket 0\otimes\ket{\Phi^+}$ and $\ket 1\otimes\ket{\Phi^+}$:
\\
$$\,\Omega (\ket 0\otimes\ket{\Phi^+})=(\bI\otimes\Sigma)(\Sigma\otimes\bI)\left[\frac{1}{\sqrt 2} (\ket{000}+\ket{011})\right]=\qquad\qquad\qquad\qquad\qquad\qquad\qquad\qquad\quad$$
$$\,\,\qquad\qquad\qquad= (\bI\otimes \Sigma)\left[\frac{1}{2}(\ket{000}+\ket{110}+\ket{011}-\ket{101})\right] =\frac{1}{\sqrt 2} (\ket{000}+\ket{110})=\ket{\Phi^+}\otimes\ket 0.$$

$$\,\Omega (\ket 1\otimes\ket{\Phi^+})=(\bI\otimes\Sigma)(\Sigma\otimes\bI)\left[\frac{1}{\sqrt 2} (\ket{100}+\ket{111})\right]=\qquad\qquad\qquad\qquad\qquad\qquad\qquad\qquad\quad$$
$$\,\,\qquad\qquad\qquad= (\bI\otimes \Sigma)\left[\frac{1}{2}(\ket{010}+\ket{100}+\ket{001}-\ket{111})\right]=\frac{1}{\sqrt 2} (\ket{001}+\ket{111})=\ket{\Phi^+}\otimes\ket 1.$$
\\
Hence $\Omega \left(\ket\psi \otimes \ket{\Phi^+}\right)=\,\,\,\ket{\Phi^+}\otimes \ket\psi$ by linearity.
\end{proof}

\vspace{0.3cm}

\noindent
Considering three qubits where two are maximally entangled, the action of $\Omega$ transfers the state of the first qubit to the third one entangling the other two. In particular the action of $\Sigma\otimes\bI$ creates a state where entanglement is distributed over all qubits. So we can assume to encrypt the state $\ket\psi$ in a 3-entangled state by application of $\Sigma\otimes\bI$ then one can apply $\bI\otimes\Sigma$ to decrypt $\ket\psi$. Since two qubits remain entangled they can be used to check the presence of eavesdroppers implementing an additional ordinary E91 protocol. 
\\
Suppose a client, Bob, prepares a qubit pair in the state $\ket{\Phi^+}_{bc}\in\sH_b\otimes\sH_c$ sending the qubit $b$ to Alice and keeping the qubit $c$ in his hands. Meanwhile Alice prepares a qubit in the state $\ket x_a\in\sH_a$ giving rise to the tripartite quantum state  $\ket{\epsilon_x}_{abc}:=\ket x_a\otimes\ket {\Phi^+}_{bc}\in\sH_a\otimes\sH_b\otimes\sH_c$, where $x=0,1$.

\vspace{-0.8cm}

$$
\Qcircuit @C=0.5em @R=1.5em {
 & & & & & \mbox{}& & & & &\\
& & & & & & & &  &\mbox{Alice} &\\
 &\ket x& & &\qw& \qw &\qswap & \qw&\ctrl{1}&\gate{H}&\ctrl{1}&\qw&\meter& y & & & & & & & & & & & & & &\qquad \qquad\quad\mbox{Bob}\\
  & & & &\qw &\qw  & \qswap\qwx & \qw & \targ&\qw&\targ&\qw&\qw&\qw&\qw&\qw &\qw &\qw&\qw&\qw&\qw&\qw&\qw&\qw&\qw&\qw&\qw&\qw&\qw&\qw&\qw&\qswap&\qw&\ctrl{1}&\gate{H} &\qw&\ctrl{1} &\meter& & y'\\
  & & &\qwx &\qw &\qw &\qw&\qw &\qw    &\qw &\qw &\qw &\qw&\qw&\qw&\qw &\qw&\qw&\qw&\qw &\qw&\qw&\qw&\qw&\qw&\qw&\qw&\qw&\qw & & &\qwx & & & & & & &\\
  & & & & & & & &  & & & & & & & & & & & & & &\ket 0 & & &\gate{H} &\ctrl{1}&\qw&\qw\qwx & & &\qwx & &\qwx &  & & \qwx& & & \\
  & & & & & & & &  & & & & & & & & & & & & & &\ket 0 & & &\qw&\targ&\qw & \qw&\qw &\qw&\qswap\qwx& \qw&\targ\qwx &\qw &\qw&\targ\qwx&\meter& & x
\gategroup {3}{1}{5}{13}{2.0em}{-} \gategroup {4}{23}{7}{38}{4.0em}{-} 
}
$$
\begin{center}
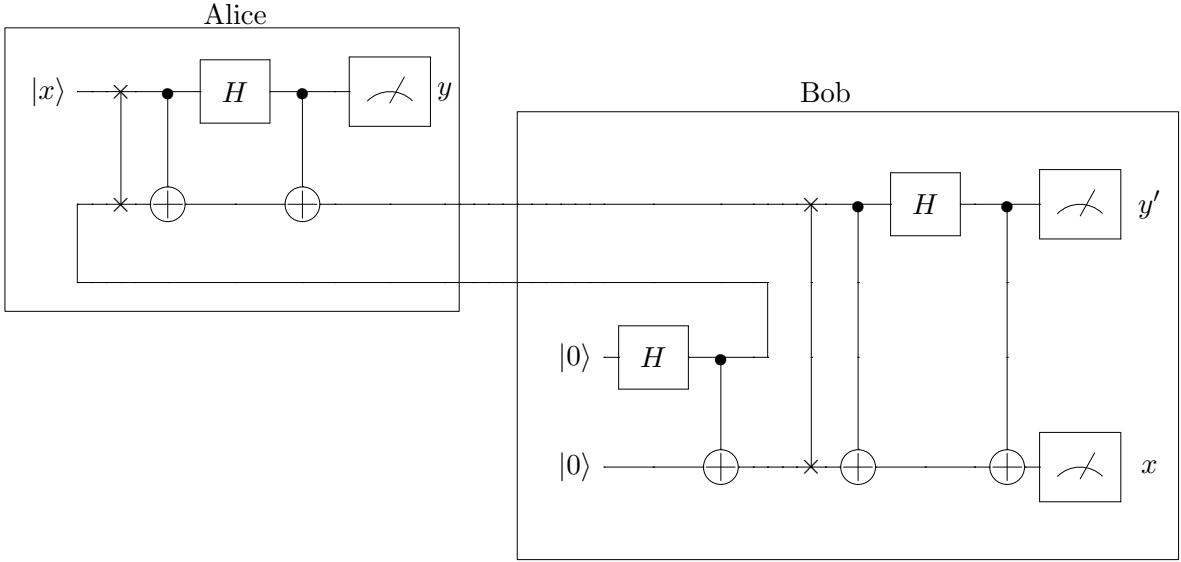
\small{Figure 1. A circuit diagram schematizing the structure of the proposed protocol.}
\end{center}

\vspace{0.5cm}

\noindent
Let $A_1, A_2, A_3, B_1, B_2, B_3$ be the observables defined in (\ref{A}) and (\ref{B}), a QKD protocol can be realized as follows:

\vspace{0.3cm}

\noindent
\textbf{Step 1}
Alice and Bob share a quantum state $\ket{\epsilon_x}_{abc}=\ket x_a\otimes\ket {\Phi^+}_{bc}$, with $x=0,1$ (the qubit pair $ab$ in Alice's lab and qubit $c$ in Bob's lab), by means of the transmission of the qubit $b$. Then Alice performs the local operation described by $\Sigma$ generating the following tripartite entangled state:
\beq\label{GHZ}
(\Sigma\otimes\bI)\ket{\epsilon_x}_{abc}=\left\{
\begin{array}{ccc}
\frac{1}{2}[\ket{000}_{abc}+\ket{110}_{abc}+\ket{011}_{abc}-\ket{101}_{abc}]& \mbox{if} & x=0\\
\\
\frac{1}{2}[\ket{010}_{abc}+\ket{100}_{abc}+\ket{001}_{abc}-\ket{111}_{abc}]& \mbox{if} & x=1
\end{array}
\right.
\eeq  

\vspace{0.1cm}

\noindent
\textbf{Step 2} Alice sends the qubit $b$ to Bob. He processes the pair $bc$ with gate $\Sigma$.
\\
\textbf{Step 3} Alice measures the observable $A_j$, with $j\in\{1,2,3\}$, on the qubit $a$. Bob measures the observable $B_j$ with $j\in\{1,2,3\}$ on the qubit $b$ and measures the observable $B_2$ on the qubit $c$ obtaining the outcome $x$. The value $x$ is stored as a bit of the secret key. The procedure is repeated $N$ times.
\\
\textbf{Step 4} Alice and Bob declare over a classical public channel the measured observables recording the outcomes produced by corresponding measurements. These bits give an additional contribution to the secret key.  
\\
\textbf{Step 5} Alice and Bob declare over a classical public channel the outcomes produced by different measurements in order to compute the test statistics (\ref{S}) detecting eventual eavesdropping attacks.
\\


\vspace{0.0cm}

\noindent
In Appendix we will remark that the state (\ref{GHZ}) belongs to the GHZ class by the computation of the Cayley's hyperdeterminant.\\ 
Assuming Steps 1-3 are repeated $N$ times ($2N$ transmitted qubits in total), Alice and Bob share a bit string of length $n\simeq N+\frac{N}{3}=\frac{4N}{3}$. Otherwise adopting the standard E91 protocol described in the previous section Alice and Bob share a key of length $n\simeq \frac{2N}{3}$ by means of the transmission of $2N$ qubits. \\
The cryptographic scheme can be represented by the quantum circuit diagram of Figure 1: The gate (\ref{bellgate}), initialized in $\ket{00}$, is inserted on the Bob's side to produce the maximally entangled state $\ket{\Phi^+}$ that is shared by the clients.

\section{Security against eavesdropping: Qualitative analysis}

Suppose an eavesdropper (Eve) wants to move an attack in order to gain the private key during the QKD. Let us show that she cannot gain sufficient information to get the sifted key \emph{nor} keeping hidden her attack. Suppose Eve is able to intercept the transmitted qubits over the quantum channel. Let us exclude the case of a \emph{man-in-the-middle} attack where Eve plays the part of Alice for Bob and viceversa sharing two different private keys with the clients because we are assuming the existence of a standard authentication protocol over the classical channel that is necessary for the sifting phase and the computation of the test statistic (\ref{S}).\\
 A general eavesdropping strategy is given by performing a measurement on the first transmitted qubit and processing the second one in a suitable way. If Eve intercepts the first qubit and performs the measurement $\{\ket 0\bra0, \ket 1\bra 1\}$ then she destroys the entanglement of the pair $bc$ shared by the clients. More in detail, suppose Eve measures the value $0$ then Alice receives the qubit $b$ in the pure state $\ket 0$ (so Bob's qubit $c$ collapses in $\ket 0$ as well and Eve can make a copy of this qubit), by the application of the gate $\Sigma$ on the separable state $\ket x\otimes \ket 0$ Alice obtains the state $\ket{\Phi^+}$ if $x=0$ or the state $\ket{\Psi^+}$ if $x=1$. Suppose Eve intercepts also the second quantum transmission (from Alice to Bob), since she 
 does not know the value of $x$, her knowledge of the total system is encoded in the following incoherent superposition:
\beq\label{eve}
\rho_{Eve}=\frac{1}{2}\left( \ket{\Phi^+ 0}\bra{\Phi^+ 0}+\ket{\Psi^+ 0}\bra{\Psi^+ 0}\right)_{}.
\eeq            
If the state of the three qubits is (\ref{eve}) then Eve cannot gain information about the value of $x$ (that is a bit of the key) by means of a single measurement process or any local operation on the qubit pair $bc'$, where $c'$ is her copy of the qubit $c$. However she makes datum $x$ inacessible to Bob as a minor result, at this stage Eve's goal is keeping hidden her interference. Therefore suppose Eve does nothing in order to avoid  any perturbation of the quantum transmission from Alice to Bob. Assuming $x=0$ (the case $x=1$ is analogue), when Bob performs $\Sigma$ on the qubit pair $bc$ then the state of the total system becomes:
\beq
(\bI\otimes\Sigma)\ket{\Phi^+ 0}=\frac{1}{\sqrt 2}\left( \ket{0\Phi^+}_{}+\ket{1\Psi^+}_{}\right),
\eeq
so when Bob measures the qubit $c$ obtains a randomized value instead of $x=0$, however the reduced state of the qubit pair $ab$ turns out to be:
\beq\label{bob}
\rho_{ab}=\frac{1}{2}\left(\ket{\Phi^+}\bra{\Phi^+}+\ket{\Psi^+}\bra{\Psi^+}\right),
\eeq
that is a separable state,  in fact it can be written as a convex combination of separable pure states: $\rho_{ab}=\frac{1}{2}(\ket{++}\bra{++}+\ket{--}\bra{--})$. As a consequence Eve's measurements on the first transmitted qubit (from Bob to Alice) can be detected by the computation of test statistic even if no operation is performed on the second transmitted qubit (from Alice to Bob). On the other hand Eve can intercept qubit $b$ once again (after Alice's processing) and perform some operation to hyde her attack during the previous quantum transmission. However any local operation performed by Eve cannot entangle the qubit $b$ (to be re-transmitted to Bob) with the qubit $c$ of Bob so the computation of test statistic reveals an interception over the quantum channel.                   
\\
Another general eavesdropping strategy is entangling the intercepted qubits with a \emph{probe quantum system} in Eve's hands. Assume Eve intercepts the first transmitted qubit and acts with an \emph{entangler}, i.e. a quantum operation on the bipartite system made by the intercepted qubit and an ancillary system (the probe) producing an entangled state. 
 So the goal of Eve is entangling a probe with Alice and Bob's systems in order to gain the values of datum $x$ (i.e. 3/4 of the secret key at the end of the round) performing measurements only on her system leaving correctly correlated Alice and Bob's qubits. This goal is unattainable because of \emph{entanglement monogamy} that can be expressed by CKW inequality \cite{ckw}: If two quantum systems $A$ and $B$ are maximally entangled then $A$ or $B$ cannot be entangled with a third system $C$. So when Alice receives the first qubit after the action of Eve's entangler, it cannot be maximally entangled with Bob's qubit as provided by the protocol then CHSH inequality cannot be maximally violated. On the other hand if Eve wants to gain the value of $x$ she must perfectly entangle the intercepted qubit with the probe destroying the quantum correlation with the Bob's qubit at all.

\section{Security against eavesdropping: Quantitative estimates}

After a qualitative discussion on intercept/re-send attacks let us quantify the security of the protocol within a general device-independent scenario. Hence we assume that any element of the QKD scheme is completely untrusted, i.e. Eve may have fabricated the measuring devices of Alice and Bob in addition to controlling the entanglement source. So Alice and Bob can estimate Eve's information only by qubit error rate $Q$ and violation $\mathcal S$ of CHSH inequality.
\\
The qubit error rate is given by  the weighted sum of the probabilities of mismatching outcomes:
\beq
Q=\frac{1}{8}p(a_2\not = b_1)+\frac{1}{8}p(a_3\not =b_2)+\frac{3}{4}p(x\not =x_b),
\eeq 
where $a_i$ are the measurement outcomes of $A_i$ obtained by Alice, $b_j$ are the outcomes of $B_j$ measured on qubit $b$  and $x_b$ is the outcome of $B_2$ measured on qubit $c$ by Bob. Then normalized mutual information between Alice and Bob is:
\beq\label{info}
I_{AB}=1-H(Q),
\eeq
where $H$ is the binary entropy function.\\
In order to eavaluate the secret key rate of the protocol, we can consider the purified version of the protocol applying lemma 1 in \cite{beaudry}. When Alice receives the qubit from Bob, she acts with an encoding (given by the processing with the gate $\Sigma$ for a fixed value of the parameter $x$) before re-sending the qubit. In the purified setting, the encoding operation (described by the action a completely positive trace-preserving map) is equivalent as a POVM-measurement (with binary outcome) performed on the received state with the half of an entangled state so that, after the measurement, the other half corresponds to the encoded state to be sent to Bob. Hence, within the purified version of the protocol, we can assume that a pure state $\ket{\Psi}_{ABE}\in\sH_A\otimes\sH_B\otimes\sH_E$ is shared among Alice, Bob and Eve and dimensions of $\sH_A$ and $\sH_B$ are unknown to Alice and Bob but fixed by Eve according to the device-independent assumption. Alice performs a measurement process $M_A$ (representing the encoding) and another one $M(A_j)$ (corresponding to the measurement of $A_j$) on two subspaces of $\sH_A$, on the other side Bob performs his measurement $M(B_j)\otimes M(B_2)$. All the measurement processes are defined by Eve who controls the devices. The measurement $M_A$ produces two possible outcomes that are assumed to be correlated to the outcomes of $M(B_2)$ then the purified protocol is equivalent to E91 protocol where the involved measurement processes are $E^A_j:=M_A\otimes M(A_j)$ and $E^B_j:=M(B_j)\otimes M(B_2)$, for $j=1,2,3$.
\\
 In order to estimate the accessible information to Eve we consider the Holevo quantity $\chi_{EB}$ between Eve and Bob. Since our purified protocol turns out to be a modification of E91 protocol within the device-independent scenario, then $\chi_{EB}$ is completely characterized by the violation $\mathcal S$  of CHSH inequality and applying the main result of \cite{acin} we have: 
\beq\label{holevo}
\chi_{EB}\leq H\left(\frac{1+\sqrt{(\mathcal S^/2)^2-1}}{2}\right).
\eeq
The secret key rate $r$ is lower-bounded by Devetak-Winter rate $r_{DW}$ \cite{dev}: 
\beq
r\geq r_{DW}=I_{AB}-\chi_{EB},
\eeq 
where $I_{AB}$ is given by (\ref{info}) and $\chi_{BE}$ is constrained by (\ref{holevo}). Under assumption of device-independence, we can provide the following secret key rate for the presented protocol:
\beq
r\geq 1-H(Q)- H\left(\frac{1+\sqrt{(\mathcal S^/2)^2-1}}{2}\right).
\eeq
\\
According to the first analysis of the protocol in section 4, the ratio between the length of the secret key and the total number of the transmitted qubit is:
\beq
R=\frac{2}{3}\,r,
\eeq   
which attains the maximum for $Q=0$ and $\mathcal S=2\sqrt 2$.\\
Summarizing: The two-way protocol can be purified applying  lemma 1 in \cite{beaudry}, so that the scheme can be described within a device-independent scenario where Alice and Bob share a pure state with Eve. They perform local measurements on their states like in E91 protocol, under collective attacks (i.e. Eve moves the same attack to each system of Alice and Bob) the Holevo quantity between Bob and Eve can be evaluated in terms of violation of CHSH inequality.

\vspace{0.0cm}

\section{Conclusions and open issues}
In this paper we have described a two-way QKD protocol based on the preparation and processing of tripartite GHZ-type entangled states. A quantum circuit has been defined to entangle three qubits in order to implement a quantum transmission where the receiver measures two distinct outcomes: The value of a fixed bit and another value to compute the test statistic to quantifying violation of CHSH inequality like within the ordinary E91 protocol. The maximum violation shows that no eavesdropping attack occurred like in the standard picture. In the ideal case of no eavesdropping  and no noise at all the length of the sifted key is $n\simeq\frac{2}{3}N$, where $N$ is the total number of transmitted qubits, instead of $n\simeq\frac{N}{3}$. 
Moreover we have focused on the security of the QKD scheme with a qualitative and a semi-quantitative analysis. Eve cannot perform any operation to gain the values of the datum $x$ or disturb the key distribution without introducing local realism in the system. The monogamy of entanglement prevents any perfectly hidden eavesdropping attack which is moved entangling intercepted qubits with an auxiliary system. In order to evaluate the trade-off between secret key rate and noise we have considered the purification of the protocol which allow us to analyze the security within a device-independent scenario where Alice, Bob and Eve share a pure state and the only figures of merit to estimate the secret key rate are the qubit error rate and the violation of CHSH inequality. Altough we presented a lower bound for the secret key rate in a device-independent scenario, further work is needed to reach a complete and satisfactory security proof of the protocol. In this regard an interesting direction of investigation could be the estimation of the effects of noise, modeled by depolarizing channels and decoherence channels, on the security.

\begin{appendix}

\noindent


\end{appendix}

\vspace{0.0cm}

\section*{Appendix}

As a technical observation let us remark that the entangled states $(\Sigma\otimes \bI_2)\ket{\epsilon_0}$ and $(\Sigma\otimes \bI_2)\ket{\epsilon_1}$, defined in (\ref{GHZ}), belong to the same entanglement class, in fact they are related by means of the invertible local operation $\bI_2\otimes\sigma_x\otimes\bI_2$. 

\begin{proposition}
The pure state (\ref{GHZ}) belongs to the entanglement class of the state GHZ.
\end{proposition}
\begin{proof}
If we consider a pure state $\ket\Psi =\sum_{i=0}^1 a_{ijk}\ket{ijk}$ of three qubits, the Cayley's hyperdeterminant of its coefficients hypermatrix is defined as:
$$
Det(\ket\Psi):=a_{000}^2a_{111}^2+a_{001}^2a_{110}^2+a_{100}^2 a_{011}^2
-2[a_{000}a_{001}a_{110}a_{111}+a_{010}a_{101}a_{101}a_{111}+\qquad\qquad\quad$$
$$\qquad+a_{000}a_{011}a_{100}a_{111}+a_{001}a_{010}a_{101}a_{110}+
a_{001}a_{011}a_{101}a_{100}+a_{010}a_{011}a_{101}a_{100}]+$$
\vspace{-0.2cm}
$$+4[a_{000}a_{011}a_{101}a_{100}+a_{001}a_{010}a_{100}a_{111}].\qquad\qquad\qquad\qquad\qquad\qquad\qquad$$
\\
\noindent
We have $Det\left[(\Sigma\otimes\bI_2)\ket{\epsilon_x}_{abc}\right]=1$ for $x=0,1$. The Cayley's hyperdeterminat is invariant under the action of $SL(2,\bC)^{\otimes 3}$ \cite{hyperdet}, then the property $Det\not=0$ is invariant under the action of $GL(2,\bC)^{\otimes 3}$. Since $Det(\ket W)=0$ and $Det(\ket {GHZ})=1/4$ we can conclude that (\ref{GHZ}) is a GHZ-type state.
\end{proof}

\end{document}